\documentclass[11pt]{article}
\usepackage{amssymb}
\usepackage{eurosym}
\usepackage[usenames,dvipsnames]{color}
\usepackage[centertags]{amsmath}
\usepackage{amsthm}
\usepackage{graphics}
\usepackage{amsmath}
\usepackage{graphicx}
\usepackage{textcomp}
\usepackage{makeidx}
\usepackage{cite}
\usepackage{wrapfig}
\usepackage{graphicx}
\usepackage{setspace}
\usepackage{graphicx}
\usepackage[english]{babel}
\usepackage{graphicx}
\usepackage{subfigure}
\usepackage{subcaption}
\usepackage{xcolor}
\DeclareGraphicsExtensions{.pdf,.jpeg,.png,.jpg}
\usepackage{caption}
\usepackage{float}
\newtheorem{theorem}{Theorem}[section]

\newtheorem{algorithm}[theorem]{Algorithm}

\newtheorem{proposition}[theorem]{Proposition}

\newtheorem{remark}[theorem]{Remark}

\newtheorem{model}[theorem]{Model}
\numberwithin{equation}{section}

\begin{document}

\title{A Mathematical Model for Chemotherapy,  Immunotherapy and Virotherapy Treatments of Cancer}
\author{Tarini Kumar Dutta$^a$, Silmera A Sangma$^a$, Janice Moore$^{b}$ and  Meir Shillor$^{b}$ \\
{\small $^a$ Department of Mathematics, Assam Don Bosco University}\\
{\small Tapesia Gardens, Sonapur Assam-782402, INDIA}\\
{\small $^{b}$Department of Mathematics and Statistics, Oakland University}
\\
{\small 146 Library Drive, Rochester, MI 48309, USA} \\
{\small tarini.kumar.dutta@dbuniversity.ac.in,\quad asangmasilmera@gmail.com}\\
 {\small jhmoore@oakland.edu \quad  shillor@oakland.edu}
}

\maketitle

\begin{abstract}
We continue our study of a model for cancer treatment, constructed in Dutta et. al., 2025, by adding Virotherapy to the Chemotherapy and Immunotherapy studied there. It is a dynamical system model for the spread of cancer in healthy tissue. It allows computer experiments of various combinations of the three modalities, which cannot be performed in the laboratory or experimentally. The novelty is the addition of Virotherapy. The analysis shows that the model solutions exist, are bounded, and nonnegative on each finite time interval, thus biologically feasible.  A time-stepping algorithm is constructed and implemented, and computer simulations are presented. 
The simulations show the development of the disease under various treatment options, including a baseline case without treatment, cases for each of the three treatments separately, and some combinations of the three treatments. {\bf These simulations indicate that combinations of treatments are more effective. However, we do not consider any limitation or incompatibilities of the joint application of the three modalities, that may exist in practice.} Once validated in the field, the model can be used to design treatment schedules of combinations of the three modalities for improved outcomes. 
\end{abstract}

\vskip4pt




\section{Introduction}
\label{sec:intro}

This work constructs, analyzes, and simulates a new model for the dynamics of a cancer tumor, when chemotherapy, immunotherapy, and virotherapy are used to treat the disease. It continues our study of models for the dynamics of cancer disease, which began in \cite{DMSS25}. There, a model for combined chemotherapy and immunotherapy
treatments was constructed, analyzed, and simulated. Here, we add virotherapy to the model. This makes it necessary to distinguish between infected and uninfected cancer cells, as well as adding an equation for the dynamics of the viruses.  The main aim is to find better ways to schedule the therapy sequence.  To that end, we construct the model, and, once verified in the field, it may help with the optimal scheduling of the three combined treatments. It is a step in our study of cancer treatments by constructing models that simulate various aspects of the dynamics of cancer tumors.
The model here, in addition to its predictive aspects, can be used to gain a deeper understanding of the dynamics of the disease. Surgery, radiation therapy, chemotherapy, and immunotherapy are the most common cancer treatments currently in use. 

Cancer annually causes the deaths of millions of people around the world, and the death rate and the related morbidity have been increasing year after year. The WHO World Cancer Report \cite{WHO22} provides recent information on the prevalence and trends of cancer in the world. There has clearly been considerable progress in the treatment of cancer; however, the current limitations of the scheduling of chemotherapy, immunotherapy, and virotherapy treatments are quite noticeable, and this work addresses the dynamics of the combined modalities. It is a contribution to the growing field of mathematical modeling of the many aspects of cancer. Indeed, mathematical modeling and computer simulations are powerful tools for simulating the dynamics of biological systems, allowing for the study of `What if?' or alternative scenarios, which cannot be performed in the field. This can lead to optimization of treatments and improvement of outcomes. 

In recent years, mathematical models have become an integral part of the study of biological systems and, what is relevant here, to the dynamics of cancer, see, e.g., \cite{KKL05, SL15,Bu20, SA21, SAES14} and the host of references therein. These are powerful tools for obtaining both a mechanistic understanding of the processes involved in cancer and possible quantitative predictions of treatment outcomes. In this way, mathematical modeling contributes to the understanding of underlying mechanisms and may provide quantitatively validated predictions.

The monograph \cite{SL15} provides a rather comprehensive description of the tools and methods, mainly from the theories of dynamical systems and optimal control, for the analysis of various population dynamics models for cancer and the tumor micro-environment, using different anticancer therapies. The analysis of such models captures well the essence of the underlying biology and sheds light on more general issues, and, in some cases, leads to predictions that are confirmed by experimental studies and clinical data.  An extensive survey of optimization of mathematical models for chemotherapy treatments is presented in \cite {SAES14}, where a considerable list of references, up to $2014$, can be found. It summarizes the various mathematical models applied to the optimal scheduling of cancer chemotherapy. Moreover, it classifies the relevant literature with respect to modeling methods, discusses the limitations of existing research, and provides some directions for further research to optimize chemotherapy treatment scheduling.

The papers \cite{PGR06,P08} develop and analyze mathematical models for cancer growth at the cell-level and its treatments by combination of immune, vaccine, and chemotherapy.  They study the stability of the equilibrium point and provide numerical simulations of mixed chemo-immuno- and vaccine therapies using both mouse and human parameters. The authors describe situations for which neither chemotherapy nor immunotherapy alone are sufficient to control tumor growth, but a combination of therapies eliminates the tumor. In this work, we find that although each therapy is sufficient to eliminate the tumor, combined therapy is more effective because it causes a significant reduction in cost, time, and possible side effects, and also the reduction in the drug amounts, the time interval of chemotherapy, and side effects.

Another model, similar to the one constructed here, can be found in \cite{DPDS20}. There, the conditions for the local stability of all the equilibrium points and the global stability condition for the tumor-free equilibrium point are studied. In addition, some representative simulations are presented. A model for the interactive dynamics between effector-tumor-normal cells can be found in \cite{DD21}. This work is based, in part, on \cite{P08}. In addition, an optimal control problem is formulated to find the best possible ways to administer immunotherapy to eradicate cancer cells without putting the patient at any health-related risk.  Their numerical simulations show that the optimal regimens eradicate the tumor load in less time than the other treatments.

We also mention some of the mathematical models constructed recently to deal with chemotherapy or immunotherapy or a combination of both, \cite{AH09, FKH14, N22,  PSZ16,  R19, S22}, and many references therein.

Oncolytic virotherapy, also known as {\it oncolytic virus therapy,} is an emerging and promising cancer treatment modality that uses genetically modified or naturally occurring viruses to selectively target and kill cancer cells. This novel approach takes advantage of the capacity of viruses to selectively infect and multiply inside living cells. By modifying the viruses to target certain traits of tumors, virotherapy can eradicate tumors while causing the least amount of harm to nearby healthy cells. Anticancer virotherapy is composed of three primary branches: Oncolytic virotherapy,  viral vector–based vaccines, and viral vector-based gene therapy. Three main groups of viruses are : (1) dsDNA viruses (e.g) Adenoviruses, Herpesviruses, Poxviruses), (2) ssDNA viruses (+ strand or “sense”) DNA (e.g. Parvoviruses), (3) dsRNA viruses (e.g. Reoviruses), \cite{WK09}. Our interest here is in the first.  Furthermore, virotherapy can resolve some of the drawbacks of traditional therapies, including non-specific toxicity and drug resistance. 

A mathematical model was created by Pooladvand et al. \cite{pool21} to examine how viral infectivity affects the interactions between oncolytic viruses and tumor cells. The spatial and temporal dynamics of uninfected tumor cells, infected tumor cells, and free viral particles within a spherically symmetric tumor environment are described by the model using a system of partial differential equations (PDEs). The infectivity parameter, which measures the virus's capacity to be internalized by tumor cells that are not infected, is a key element in the model. This parameter takes into account the changes in the tumor microenvironment that impact viral uptake and is derived from experimental data. 

A statistical framework was created by Friedman and Lai \cite{FL18} to assess the effectiveness of oncolytic viruses in conjunction with checkpoint inhibitors, particularly anti-PD-1 antibodies, in cancer treatment. The authors built a system of partial differential equations to simulate how immune cells, tumor cells, and oncolytic viruses interact. According to the study, therapy efficacy is not always improved by raising the checkpoint inhibitor's dosage. Higher doses of the checkpoint inhibitor may paradoxically reduce the efficacy of the combination therapy in specific parameter regimes. This result emphasizes how crucial it is to precisely balance the dosages of the viruses and checkpoint inhibitors to obtain the best possible therapeutic results.

The structure of this work is as follows. Section \ref{sec2} describes our mathematical model, which is in the form of six ordinary differential equations (ODE), modifying the model constructed in \cite{DMSS25}, by adding virotherapy. This makes it necessary to split the cancer cells into infected by the viruses and uninfected, and to add a rate equation for the viruses. The model is somewhat similar to those in \cite{DPDS20,DD21}. We explain the assumptions underlying the model and the various terms in the rate equations. The system describes the dynamics of Healthy, Immune, Infected and Uninfected cancer cell populations, and of the chemo- and virotherapy. Section 3 analyzes the model, showing that the solutions exist and are bounded on every finite time interval, and when the initial conditions are nonnegative, the solutions remain nonnegative. Section 4 is short and deals with the analysis of the system's steady states and their stability. It is found that the Disease Free Equilibrium (DFE) is stable and attracting. However, our interest is in reaching the steady state as fast as possible, and not just waiting until cancer vanishes on its own, which may take a very long time. Section 5 describes an algorithm for the numerical simulations of the model, very similar to the one in \cite{DMSS25}, and then presents a few representative solutions. In particular, it shows that combined chemo-immuno-virotherapy leads to better outcomes than each one separately. The numerical converges of the algorithm is shown in Section \ref{sec:conv}. The paper concludes in Section 6, where we summarize the results and comment on some near-future research directions.

\section{Mathematical Model}
\label{sec2}

We follow the structure in \cite{DMSS25}  and modify the generic model for combined chemotherapy and immunotherapy of a cancer tumor by adding virotherapy. In this process, a chemotherapeutic drug, white platelets, and an infusion of modified viruses are administered to the tumor to eliminate it.  The viruses injected into the tumor affect the cancer cells by infecting and disabling a fraction of them.  The model allows us to study the effects of each treatment on its own and of possible combinations of treatments. {\bf However, the model is generic and we do not have field data to corroborate it. Furthermore, we do not take into account any contraindications on the use of combinations of the treatment modalities.}

The model deals with total numbers, assumed to be large, so that ODEs may be applied. We let the total number of normal or healthy cells be $H(t)$, immune cells $I(t)$,  tumor cells infected by viruses $K_I(t)$, uninfected tumor cells $K_U(t)$, the amount of chemotherapy drug $D(t)$ and the number of viruses $V(t)$, all in the tumor and its immediate surroundings,  at time $t$. Whereas the drug affects all cells, the viruses are very specific and affect only cancer cells. {\bf The populations are measured by the number of individuals, while the drug by arbitrary units.}

We measure time in days, and the time interval of interest is $0\leq t \leq T$, for some $T>0$. 

The proposed model deals with the rate of change of the six variables as follows.

\begin{model}[{\it Model for Chemotherapy, Immunotherapy and Virotherapy of a Cancer Tumor}]
\label{model1} 
Find six functions  $H, I, K_I, K_U, D,V:[0, T]\to \mathbb{R}^{+},$ such that:
\begin{eqnarray}
\frac{dH}{dt}&=&\alpha_1 H(1-\alpha_H H) -\gamma_1 K_UH-\mu_1 DH-\mu_{H}H,
\label{21}\\
\frac{dI}{dt} &=& q+\psi+\alpha_2I(1-\alpha_I I) - \gamma_2 K_UI -\mu_2 DI -\mu_{I}I, \label{22}\\
\frac{dK_I}{dt} &=& \frac{\delta_1 K_U V}{\delta_2 +K_U}  -\delta_3 K_I I -\delta_4 K_I D
-\mu_{K_I}K_I, \label{23}\\
\frac{dK_U}{dt} &=&\alpha_3 K_U(1-\alpha_{K_U} (K_I+K_U))\chi_1(K_U)-\frac{\delta_1K_UV}{\delta_2+K_U} -\gamma_3 K_UI-\gamma_4 K_UH \nonumber \\
&-&\mu_3 K_UD -\mu_{K_U}K_U, \label{24}\\
\frac{dD}{dt} &=&\varphi_D-\beta_{1} ID-\beta_{2}HD-\beta_{3}(K_U+K_I)D-\mu_{D} D, \label{25}\\
\frac{dV}{dt} &=&\varphi_V + \beta_{4}(\delta_3 I +\delta_4 D)K_I -\mu_{V} V.  \label{26}
\end{eqnarray}
Together with the initial conditions,
\begin{equation}
\label{27}
H(0)=H_0,\; I(0)=I_0, \; K_I(0)=K_{I0},  \; K_U(0)=K_{U0},\; D(0)=D_0,\; V(0)=V_0.
\end{equation}
\end{model}
The rates are per day; the populations $H, I, K_I, K_U$ are measured in the number of individuals; while the units of $D, V$ are arbitrary; and the various terms have the correct units, so that $\beta_3K_UD$ or $\varphi_D$ have units of drug/day, while $\varphi_V$ has units of particles/day, and $q$ and $\psi$ have units of cells/day.

Equation (\ref{21}) describes the rate of change of normal cells, which is assumed to grow logistically, with a growth rate $\alpha_1$ and carrying capacity ${1}/{\alpha_H}$. The term $\gamma_1 K_U H$ represents the loss rate due to interaction with uninfected tumor cells;  $\mu_1 DH$ is the death rate caused by the drug;  and $\mu_{H}H$ is the natural death rate of the cells in the absence of the chemo drug.  It is assumed that viruses are very specific and do not affect healthy cells.

Equation (\ref{22}) describes the rate of growth of immune cells in the tumor and its vicinity. The natural supply rate of immune cells from internal growth and from outside or the rest of the body is $q=q(t)$, and $\psi=\psi(t)$ represents the immunotherapy process in which white platelets are added to the tumor, thus increasing the number of immune cells.  The tumor-specific immune response is also given by the logistic term, with growth rate $\alpha_2$ and carrying capacity ${1}/{\alpha_I}$. The term $\gamma_2K_U I$ describes the rate at which immune cells die due to their interaction with tumor cells, and $\mu_2 DI$ is the death rate caused by the drug. Finally,  $\mu_{I} I$ represents the natural mortality rate.

As noted in \cite{DMSS25},   we separated $q$ and $\psi$ because the first is an intrinsic growth and the second represents the medical intervention.

The rate of change in the number of infected tumor cells is described by equation (\ref{23}). The growth rate, caused by infection, is assumed to be the rate at which the uninfected tumor cells become infected, $\frac{\delta_1 K_U V}{\delta_2 +K_U}$, where $\delta_1$ denotes the growth rate constant and $\delta_2$ a weight factor. The second and third terms on the right-hand side, $\gamma_3 K_II$, $\delta_4 K_ID$, represent the loss rate caused by interactions with immune cells and the rate at which the drug kills infected tumor cells, respectively. The last term describes their natural death rate, where the rate constant $\mu_{KI}$ is likely to be very small. 

 Equation (\ref{24}) describes the rate of change in the number of active uninfected tumor cells, the ones that are the most dangerous. The growth rate is assumed to be logistic, $\alpha_3 K_U(1-\alpha_{KU} (K_I+K_U))\chi_1(K_U)$, where $\alpha_3$ denotes the growth rate constant and $\alpha_{KU}$ the reciprocal of the carrying capacity. The introduction of $\chi_1(K_U)$ is for the reasons given below. The second term on the right-hand side, $\frac{\delta_1 K_UV}{\delta_2+K_U}$, describes the rate at which the viruses infect the uninfected cancer cells, with the rate constant $\delta_1$, which is likely to be small, and $\delta_2$ is a weight factor. The terms $\gamma_3 K_U I$, $\gamma_4 K_U H$ and $\mu_3 K_U D$ describe the interactions between the uninfected cancer cells and the immune cells, the healthy cells, and the chemo drug, respectively.  The last term describes their natural death rate, where $\mu_{KU}$ is likely to be very small. 
 
 {\bf We note that  the term $\gamma_1 K_UH$ in (2.1) describes the interaction of the healthy cells with the uninfected cells, competing for physical space and resources in the tumor. Similarly are the terms $\gamma_3 K_UH$ in (2.4), and $\gamma_2 K_UI$ in (2.2). It is very likely that $\gamma_1, \gamma_2, \gamma_3$ are very small, however, the terms were added for the sake of completeness.}

For mathematical reasons, we introduce the {\it cutoff function} $\chi_1$, which is the Heaviside function, defined as
\begin{equation}
\label{chi}
 \chi_1(K_U)=
 \begin{cases}
  0&  \quad 0\leq K_U \leq 1 \\
  1&\quad 1 <K_U.
 \end{cases}
\end{equation}
The addition of this cutoff function guarantees that when $K_U<1$, the cancer vanishes and cannot reappear. Indeed, we have the following remark.

\begin{remark}
The logistic term in (\ref{24})  makes biological sense in describing the growth of cancer cells up to the carrying capacity of the system, but it has a mathematical drawback, since for exceedingly small $K_U<<1$, once drug administration is stopped, the cancer rebounds and $K_U$ grows to its carrying capacity, which defeats the purpose of the model. Therefore, we introduce the cutoff function $\chi_1$, which turns off the growth term once the population of cancer cells is below $1$. 
\end{remark}

Equation (\ref{25}) describes the dynamics of the chemotherapy drug.  The first term, $\varphi_D =\varphi_D(t)$, represents the rate at which the drug is administered. The drug decreases because of the interactions with immune, healthy and cancer cells. The rates are given by $ \beta_{1}ID$, $ \beta_{2}HD$, and $\beta_{3}(K_U+K_I)D$, respectively. Here, we assume that the chemo drug, by killing an infected cell, releases the viruses in it.  This is taken into account with the $\delta_4D$ term in equation \ref{26}. It is likely that $\delta_4$ is very small. The natural decay of the drug in vivo is $\mu_{D} D$.  It is assumed that the chemotherapy drug kills all cell types, but at different rates.

Next, (\ref{26}) is the rate equation for virotherapy and its effects on the tumor. Similarly to \cite{Khan22}, cancer cells are divided into infected and uninfected. The term $\varphi_V =\varphi_V(t)$ is the rate at which virotherapy, in the form of viruses, is administered to the tumor.  The term $\beta_4(\delta_3  I +\delta_4 D)K_I$ describes the generation of new viruses when infected cancer cells burst, either from the immune cells or chemotherapy drug, and die, releasing new viruses. The factor $\beta_{4}$ is the average number of viruses released from each bursting infected cell. The last term describes the decrease in the number of viruses due to other causes.

We note that the amounts of chemotherapy drug and virotherapy are represented by two equations, but we add the immunotherapy into the equation for the immune cells, since it seems more natural to consider them as part of the immune system.

To complete the model, we prescribe in (\ref{27}) the initial cell populations, the initial amount of the drug, and the initial amount of virotherapy. For the sake of biological sense, we assume that
\[
H(0)=H_0\in (0, 1/\alpha_H),\; I(0)=I_0>0,\; K_I(0)=K_{I0}\geq  0, \; K_U(0)=K_{U0}> 0,
\]
\begin{equation}
\label{29}
 D(0)=D_0\geq 0,\quad V(0)=V_0\geq 0.
\end{equation}

 We note that the model is rather complex, with nonlinear interactions and $29$ rate constants, four time-dependent input functions $q, \varphi_D, \varphi_V$ and $\psi$, and six initial conditions.

\section{Existence, positive invariance and boundedness}
\label{sec:3b}
We perform a mathematical analysis and prove that a solution for the model exists, and it is positive and bounded. Since the right-hand sides are locally Lipschitz continuous, the system (\ref{21})--(\ref{26}) has a local solution. Using the initial conditions (\ref{27}) it follows that there exists $0<t_0$ such that the solution exists on the time interval $[0, t_0]$, it is nonnegative and bounded. Let $C^*$ denote the bound of all the variables on $[0, t_0]$.

To establish global solvability, we must show that the model is biologically feasible, meaning that if we start with nonnegative initial conditions, all the solutions remain nonnegative.  We also show that every solution is bounded on every arbitrary finite time interval. This, in turn, implies the global solvability on every finite time interval.

\begin{proposition}
\label{prop3.1}  Assume that the initial conditions satisfy (\ref{27}) and $q,\varphi_D, \varphi_V$ and $\psi$ are Lipschitz, nonnegative and  bounded functions, say $0\leq q, \varphi_D, \varphi_V, \psi \leq M$. Then, every solution of system  (\ref{21})--(\ref{26}) satisfies the estimates:
\begin{eqnarray}
&& 0  <  H \leq      \frac{H_{0}e^{\alpha_1 T}}{1-\alpha_H H_0+H_{0}\alpha_He^{\alpha_1 T}}\leq \frac{1}{\alpha_H},    \label{31}\\
&&  0  <  I  \leq \frac{\alpha_2 S+1}{2\alpha_I S} ,     \label{32}\\
&&  0  \leq   K_I \leq  \frac{A_*}{\mu_{K_I}}+K_{I0},    \label{33}\\
&&  0  \leq   K_U \leq     \frac{K_{U0}e^{\alpha_3 T}}{(1-\alpha_{K_U} K_{U0})+K_{U0}\alpha_{K_U}e^{\alpha_3 T}}\leq \frac{1}{\alpha_{K_U}} ,    \label{34}\\
&& 0  \leq D\leq  \frac{M}{\mu_D}+  \left(D_{0}-\frac{M}{\mu_D}\right)e^{-\mu_DT}   \leq \frac{M}{\mu_D}+D_0,  \label{35}\\
&& 0  \leq V \leq  \frac{C_*}{\mu_V}+  \left(V_{0}-\frac{C_*}{\mu_V}\right)e^{-\mu_VT}   \leq \frac{M}{\mu_V}+V_0,\label{36}
\end{eqnarray}
 for $0<T<\infty $. Here,
 \begin{align}
  S&=\frac{\alpha_2(1-2\alpha_I I_0) +\Delta}{\alpha_2(1-2\alpha_I I_0) -\Delta},
  \label{37}\\
  \Delta^2&=\alpha_2^2 +8\alpha_2\alpha_IM, \label{38}\\
  C_*&=M+\beta_4(\delta_3+\delta_4)(C^{*})^{2}, \label{39}\\
  A_*&=\frac{\delta_1}{\delta_2}(C^{*})^{2}. \label{310}
 \end{align}
\end{proposition}

We conclude that the {\it feasible region} $\Omega$ is
\[
\Omega=\left\{ (H, I, K_I, K_U, D, V)\in \mathbb{ R}_+^6 : 0\leq   H\leq \frac{1}{\alpha_H},\quad I(t)\leq \frac{\alpha_2 S+1}{2\alpha_I S}, \right.
\]
\[
 \left. \quad K_I\leq \frac{A_*}{\mu_{K_I}}+K_{I0},
 \quad K_U\leq \frac{1}{\alpha_{KU}},\quad D\leq \frac{M}{\mu_d}+  D_{0},\quad V(t)\leq \frac{C_*}{\mu_V} +V_0 \right\},
\]
which is positively invariant for the system (\ref{21})-(\ref{26}).

\noindent
\begin{proof}
 First, since the initial conditions are positive (except for $K_I$, $K_U$, and $V$, which we may choose as small and positive, and then pass to the limit) by the continuity of the solutions, there is a time interval $[0, t_0]$, $t_0>0$, on which the solutions are nonnegative and bounded by some positive constant $C^*$.

 It follows from (\ref{21}) that
\[
\frac{dH}{dt}\leq \alpha_{1}H(1-\alpha_H H).
\]
Let $\tilde H$ be the solution of the logistic equation 
\[
\frac{d \tilde H}{dt}= \alpha_{1}\tilde H(1-\alpha_H \tilde H),\qquad \tilde H(0)=H_0,
\]
then $H\leq \tilde H$. It is known, see, e.g, \cite{BoDi97}, or by using symbolic computations, that
 \[
 \tilde H=\frac{H_{0}}{(1-\alpha_H H_0)e^{-\alpha_1 T}+H_{0}\alpha_H}
 \leq \frac1{\alpha_H}.
 \]

Next, since $H_{0}>0$, the model is continuous and the solutions are locally bounded, (\ref{21}) indicates that
\[
H'>-(\alpha_1\alpha_H C^{*} +\gamma_1C^* +\mu_1 C^* +\mu_H)H\geq -c_1H,
\]
where $c_1>0$, and this implies that
\[
H(t)\geq H_0e^{-c_1 t}>0.
\]
 This establishes the bounds in (\ref{31}).  Estimate (\ref{34}) is derived similarly,
 when we note that
 \[
 \alpha_3 K_U(1-\alpha_{K_U} (K_I+K_U))\chi_1(K_U)\leq \alpha_3 K_U(1-\alpha_{K_U} K_U).
 \]
 The cutoff function $\chi_1$ does not change the estimate.

To obtain (\ref{32}) we proceed in the same way. Equation (\ref{22}) implies that
\[
I'\leq q + \psi +\alpha_2 I(1-\alpha_I I)\leq 2M+\alpha_2 I(1-\alpha_I I),
\]
hence, 
\[
\frac{dI}{2M+\alpha_2 I(1-\alpha_I I)}\leq dt.
\]
Integrating and using (\ref{37}) and (\ref{38}), see e.g. \cite[Formula 3.3.17]{AS70}, yields
\[
\frac1\Delta \log\left(\frac{\alpha_2(1-2\alpha_I I) -\Delta}{\alpha_2(1-2\alpha_I I) +\Delta}\right)\leq t+C_2,
\]
where $C_2$ is an integration constant, determined by the initial condition $I(0)=I_0$.
Now, tedious and long, but straightforward manipulations lead to
\[
I(t)\leq \frac{\Delta -\alpha_2+e^{\Delta t}( S\alpha_2 +1)}{(2\alpha_I)(-1+Se^{\Delta t})}\leq \frac{\alpha_2 S+1}{2\alpha_I S}.
\]
To obtain the lower bound, we note that
\[
I' \geq  -(\alpha_2\alpha_1C^*+\gamma_2C^*  + \mu_2 C^* +\mu_I )I\geq -c_{2}I,
\]
thus,
\[
I(t) \geq I_0 e^{-c_{2}t}>0.
\]
These inequalities establish (\ref{32}).

It follows from (\ref{23}), since $V\leq C^*$ and using (\ref{310}), that 
\[
\frac{d K_I}{dt}\leq \frac{\delta_1 K_UC^*}{\delta_2 + K_U}-\mu_{K_I}K_I\leq A_*-\mu_{K_I}K_I,
\]
therefore,
\[
\frac{d K_I}{A_*-\mu_{K_I}K_I}\leq dt.
\]
Then, simple manipulations using the fact that $K_I(0)=K_{I0}$, yield
\[
K_I(t)\leq \frac{A_*}{\mu_{K_I}} +K_{I0}e^{-\mu_{K_I} t}\leq \frac{A_*}{\mu_{K_I}} +K_{I0}.
\]
Furthermore,
\[
K_I'\geq -(\delta_3C^*+\delta_4C^*+\mu_{K_I})K_I\geq -c_4K_I,
\]
and since this implies
\[
K_I>K_{I0}e^{-c_4T}> 0,
\]
we have (\ref{33}).

Equation (\ref{25}) satisfies 
\[
\frac{dD}{dt}\leq M-\mu_D D.
\]
It follows from \cite{BoDi97}, or using software, that
\[
D(t)\leq \frac{M}{\mu_D} +\left(D_0 -\frac{M}{\mu_D} \right) e^{-\mu_D t}.
\]
The methods used for the previous estimates show that $D(t)\geq 0$. This establishes (\ref{35}).
Equation (\ref{26}) shows that 
\[
\frac{dV}{dt}\leq M + \beta_4(\delta_3+\delta_4)(C^*)^2 -\mu_V V=C_*-\mu_V V. 
\]
It follows from e.g. \cite{BoDi97}, or using software, that
\[
V(t)\leq \frac{C_*}{\mu_V }+\left(V_0- \frac{C_*}{\mu_V }\right)e^{-\mu_V t}.
\]

For the lower bound of $V$, we obtain
\[
V'\geq -\mu_V V.
\]
Similar to the previous equations, it follows that $V\geq 0$.
This completes the proof.
\end{proof}

This result shows that, in addition, the system is Lipschitz continuous on every arbitrary time interval. Therefore, we have the following global existence result.
\begin{theorem}[Global Existence]
 Assume that the initial conditions satisfy (\ref{27}) and $q,\psi,\varphi_D $ and $\varphi_V$ are smooth, positive, and bounded functions. Then, system  (\ref{22})--({2.5}) has a unique smooth solution on every  time interval $[0, T]$, $0<T\leq \infty$.   
\end{theorem}
 
\section{Equilibrium points and their stability}
\label{sec4}
 Equilibrium points are found by equating the derivatives to zero. However, in our case, the equilibrium points are only of mathematical interest, since the purpose of the model is to find ways to drive the system as quickly as possible to the DFE, the state without cancer.
 
 Let $(\bar{H},\bar{I}, \bar{K_I},  \bar{K_U}, \bar{D}, \bar{V}) $ denote the equilibrium values. The equilibrium system  (omitting the bars) is the following:
\begin{eqnarray}
0&=&\alpha_1 H(1-\alpha_H H) -\gamma_1 K_UH-\mu_1 DH-\mu_{H}H,\notag \\
0 &=& q+\psi+\alpha_2I(1-\alpha_I I) - \gamma_2 K_UI -\mu_2 DI -\mu_{I}I,\notag  \\
0 &=& \frac{\delta_1 K_U V}{\delta_2 +K_U}  -\delta_3 K_I I -\delta_4 K_I D
-\mu_{K_I}K_I, \notag \\
0  &=&\alpha_3 K_U(1-\alpha_{K_U} (K_I+K_U))\chi_1(K_U)-\frac{\delta_1K_UV}{\delta_2+K_U} -\gamma_3 K_UI-\gamma_4 K_UH  \nonumber \\
&-&\mu_3 K_UD -\mu_{K_U}K_U, \label{41}  \\
0 &=&\varphi_D-\beta_{1} ID-\beta_{2}HD-\beta_{3}(K_U+K_I)D-\mu_{D} D, \nonumber\\
0 &=&\varphi_V + \beta_{4}(\delta_3 I +\delta_4 D)K_I -\mu_{V} V.  \nonumber
\end{eqnarray}

For biological reasons, we assume that
the growth rate constants are larger than the death rate constants, 
\begin{equation}
\label{42}
\alpha_1>\mu_H,\qquad  \alpha_2>\mu_I.   
\end{equation}

To study the stability of the equilibrium points, we compute the
Jacobian matrix \( J \) of the system, given by:
\[
J(H, I, K_I, K_U, D, V)  =
\]
\begin{equation}
\label{J}
\begin{bmatrix}
J_{11} & 0 & 0& -\gamma_1 H & -\mu_1 H &0 \\
0 &J_{22}& 0&-\gamma_2 I & -\mu_2 I&0 \\
0& -\delta_3K_I&J_{33}& \frac{\delta_1\delta_2 V}{(\delta_2+K_U)^2}&-\delta_4K_I&\frac{\delta_1K_U}{\delta_2+K_U}\\
-\gamma_4K_U& -\gamma_3 K_U& -\alpha_3\alpha_{K_U}K_U\chi_1 & J_{44}& -\mu_3 K_U & -\frac{\delta_1 K_U}{\delta_2+K_U} \\
-\beta_2 D & -\beta_1 D & -\beta_3 D & -\beta_3 D&J_{55}& 0\\
0&\beta_4\delta_3K_I &\beta_4\delta_3I+\beta_4\delta_4D& 0&\beta_4\delta_4K_I& -\mu_V
\end{bmatrix},
\end{equation}

where the terms $ J_{11}-J_{55}$ are given by:
\begin{eqnarray*}
J_{11} &=& \alpha_1 (1 - 2\alpha_H H) - \gamma_1 K_U - \mu_1 D - \mu_H,\\
J_{22} &=& \alpha_2 (1 - 2\alpha_I I) - \gamma_2 K_U - \mu_2 D - \mu_I,\\
J_{33} &=&     - \delta_3 I - \delta_4 D-\mu_{KI} ,\\
J_{44} &=& (\alpha_3 -\alpha_3\alpha_{KU}K_I-2\alpha_3\alpha_{KU}K_U)\chi_1
-\frac{\delta_1\delta_2V}{(\delta_2+K_I)^2} -\gamma_3I-\gamma_4H-\mu_3D-\mu_{KU},\\
J_{55}&=& -\beta_2H-\beta_1I-\beta_3(K_I+K_U)-\mu_D.
\end{eqnarray*}

\subsection{Disease free equilibrium}
\label{sec:dfe}

First, we study the DFE and observe that in such a case we need to have $K_I=K_U=0$ and $\psi=\varphi_D=\varphi_V=0$, hence $D=V=0$, since neither drug nor virotherapy is needed. Moreover, $q=const. >0$. Thus, the system reduces to
\[
\alpha_1 H(1-\alpha_H H)-\mu_{H}H=0,
\] 
\[
q+\alpha_2 I(1-\alpha_I I) -\mu_{I}I=0
\]
The first equation has two solutions $H_1=0, H_2=(\alpha_1-\mu_H)/\alpha_1\alpha_H$. The second equation also has two solutions:
\[
I_{1,2}=\frac{(\alpha_2-\mu_I)\pm \sqrt{(\alpha_2-\mu_I)^2+4\alpha_2\alpha_Iq}}{2\alpha_2\alpha_I}.
\]

We note that the solution $I_1$ is positive, since $\alpha_2>\mu_I$, while $I_2$ is negative, which we discard as a mathematical artifact. Indeed, since the solutions are nonnegative, Proposition 3.1, the negative solution cannot be reached as a long-time limit, starting with nonnegative initial conditions.

Therefore, we find four DFE solutions $(H_{1,2},I_{1,2})$. However, the two solutions with $H_1=0$ and those with $I_2$ are of no interest. Hence, we consider the only feasible solution

\begin{equation}
\label{44}
\left(\frac{\alpha_1-\mu_H}{\alpha_1\alpha_H}, I_1\right).    
\end{equation}

In the numerical simulations, in the case of convergence to a DFE, we have
\[
(H, I)= (1000,500). 
\]

Then, in the DFE, we have
\begin{equation}
\label{45}
J(H, I, 0, 0, 0, 0)  =
\begin{bmatrix}
J_{11} & 0 & 0& -\gamma_1 H & -\mu_1 H &0 \\
0 &J_{22}& 0&-\gamma_2 I & -\mu_2 I&0 \\
0& 0&J_{33}&0&0&0\\
0& 0& 0 & J_{44}& 0 &0 \\
0 & 0 &0 & 0&J_{55}& 0\\
0&0&J_{63}& 0&0&-\mu_V
\end{bmatrix},
\end{equation}
Here, we used the fact that $\chi_1(0)=0$, and 
\[
J_{11} = (\alpha_1 - \mu_H)  - 2\alpha_1\alpha_H H=-(\alpha_1 - \mu_H)<0,
\]
\[
J_{22} = (\alpha_2  - \mu_I) - 2\alpha_2\alpha_I I=
-\sqrt{(\alpha_2-\mu_I)^2+4\alpha_2\alpha_Iq}<0,
\]
\[
J_{33} =  - \delta_3 I -\mu_{KI} =-\delta_3\left(\frac{(\alpha_2-\mu_I)+ \sqrt{(\alpha_2-\mu_I)^2+4\alpha_2\alpha_Iq}}{2\alpha_2\alpha_I}\right)-\mu_{KI}<0,
\]
since $\chi_1(0)=0$, we have

\[
J_{44} = -\gamma_3 I - \gamma_4 H -\mu_{KU}=-\gamma_3\left(\frac{(\alpha_2-\mu_I)+ \sqrt{(\alpha_2-\mu_I)^2+4\alpha_2\alpha_Iq}}{2\alpha_2\alpha_I}\right)
\]
\[
-\gamma_4\frac{(\alpha_1-\mu_H)}{\alpha_1\alpha_H}-\mu_{KU} 
<0,
\]
\[
J_{55}=-\beta_2H-\beta_1 I-\mu_D=-\beta_1\left(\frac{(\alpha_2-\mu_I)+ \sqrt{(\alpha_2-\mu_I)^2+4\alpha_2\alpha_Iq}}{2\alpha_2\alpha_I}\right)
\]
\[
-\beta_2\frac{(\alpha_1-\mu_H)}{\alpha_1\alpha_H}-\mu_{D}  <0,
\]
\[
J_{63}=\beta_4\delta_3I=\beta_4\delta_3\left(\frac{(\alpha_2-\mu_I)+ \sqrt{(\alpha_2-\mu_I)^2+4\alpha_2\alpha_Iq}}{2\alpha_2\alpha_I}\right)>0.
\]

It is seen that all the terms on the main diagonal are negative, and so are all the terms above the main diagonal. Using Mathematica, we find that the eigenvalues are
\begin{equation}
\label{46}
\lambda_1=J_{11},\;\; \lambda_2=J_{22},\;\;\lambda_3=J_{33},\;\; \lambda_4=J_{44},\;\;
\lambda_5=J_{55},\;\;\lambda_6=-\mu_V.
\end{equation}
Since all the eigenvalues are negative, we conclude,
\begin{proposition}
  The Disease Free Equilibrium in model (\ref{41}) is stable and attracting.  
\end{proposition}

However, as was noted above, our interest is in reaching the disease-free state as fast as possible, and not just waiting until cancer vanishes on its own, which may take a very long time.

\subsection{Endemic equilibrium}
\label{sec:ee}

To obtain the EE states, we need to solve the full system (\ref{41}), with $\chi_1=1$.
However, our main interest in this work is in the eradication of cancer cells. Hence, the EE are only of mathematical interest.  For the sake of completeness, we describe the stability of the EE
in the Appendix.

\section{Numerical algorithm and simulations}
\label{sec:num}
To obtain insight into the model's solutions and some quantitative results, we compute approximate solutions of the model. To that end, we constructed an explicit time-stepping algorithm and wrote the related program in MATLAB.  We present the model solutions when only chemotherapy is administered, only immunotherapy is used, only virotherapy is used, and when all three are applied.

The time interval $[0, T]$ is discretized into $N$ subintervals of length $\Delta t=T/N$. Typically, in the simulations $\Delta t=3.6698\text{ x }10^{-5}$ is used. Denoting the value of a function $f$ at time $t_n=n\Delta t$ by $f^n$, the discretized system, for $n=0,1,\dots N-1$, is as follows.

\begin{eqnarray}
H^{n+1}&=&H^n +\left(\alpha_1 H^n(1-\alpha_H H^n) -\gamma_1 K_U^nH^n-\mu_1 D^nH^n-\mu_{H}H^n\right)\Delta t ,
\label{51}\\
I^{n+1} &=&I^n+\left( q^n+\psi^n+\alpha_2 I^n(1-\alpha_I I^n) - \gamma_2 K_U^nI^n -\mu_2 D^nI^n -\mu_{I}I^n\right)\Delta t,
\label{52}\\
K_I^{n+1} &=&K_I^n+\left(\dfrac{\delta_1K_U^nV^n}{\delta_2+K_U^n}-\delta_3K_I^nI^n-\delta_4K_I^nD^n-\mu_{KI}K_I^n\right)\Delta t,
\label{53}\\
K_U^{n+1} &=&K_U^n+(\alpha_3 K_U^n(1-\alpha_{KU} (K_I^n+K_U^n))\chi_1(K_U^n)-\dfrac{\delta_1K_U^nV^n}{\delta_2+K_U^n}-\gamma_3  K_U^n I^n  
\nonumber\\
&& - \gamma_4 K_U^nH^n - \mu_3 K_U^nD^n-\mu_{KU}K_U^n)\Delta t,
\label{54}\\
D^{n+1} &=&D^n+\left(\varphi_D^n-\beta_{1} I^nD^n-\beta_{2}H^nD^n-\beta_{3}(K_U^n+K_I^n)D^n-\mu_{D} D^n\right)\Delta t, 
\label{55}\\
V^{n+1} &=&V^n+(\varphi_V^n+\beta_4(\delta_3I^n+\delta_4D^n)K_I^n-\mu_VV^n)\Delta t.
\label{56}
\end{eqnarray}
Together with the initial conditions,
\begin{equation}
\label{57}
H^0=H_0,\; I^0=I_0, \; K_{I}^0=K_{I0},\; K_U^0=K_{U0}, \; D^0=D_0,\; V^0=V_0.
\end{equation}

\newpage
Using the above discretization of the model equations, we have the following algorithm.\\
\rule{15.3cm}{.03cm}
\begin{algorithm} Explicit time-stepping\\
\rule{15.3cm}{.01cm}\\
 \hskip2cm Set $H^0, I^0, K_I^0, K_U^0, D^0$ and $V^0$ using (\ref{57}) \\
\indent \textbf{for} n = 0 to N-1 \textbf{do}\\
 \hskip2cm Calculate $H^{n+1}, I^{n+1}, K_I^{n+1}, K_U^{n+1}, D^{n+1}$ and $V^{n+1}$ according to (\ref{51})--(\ref{56})
\vskip1pt
\indent \textbf{end for}\\
\rule{15.3cm}{.03cm}
\end{algorithm}

The model is new, so we experimented with the values of the various parameters, using `reasonable' values, to obtain the different possible types of behavior of the system. Since our main interest is to eradicate cancer, we do not dwell on the EE state, nor on the various other types of behavior generated during the simulations. The parameters and coefficients used here are shown in Table \ref{tab:1}.

\subsection{Simulations}
\label{sec:dfe}

First, we consider the system without any treatment.  Using the parameter values in Table \ref{tab:1} together with the initial values $H_0 = 300$, $I_0 = 250$, $K_{I0} = 0$, $K_{U0} = 500$, $D_0 = 0$ and $V_0 = 0$ leads to the results depicted in Figure \ref{fig:NoTreatment250707}.  In the presence of cancer, and without treatment, immune cells decrease rapidly;  initially, the number of healthy cells grows and then decreases to zero, while the cancer increases to its carrying capacity of $K_U=1000$. The model predicts the spread of cancer to its maximal possible extent. 

Next, we consider each treatment on its own.  Applying only chemotherapy, we let $\varphi_D = 2000$ for 1 hour once a week. Figure \ref{fig:ChemoOnly250707} shows that the cancer is eliminated on day 28, after three chemotherapy treatments.  When cancer is eliminated, the treatment is discontinued.   Healthy and immune cells are seen to return to their steady state values, $H=1000$ and $I=500$, indicating a complete recovery from the disease. 

Figure \ref{fig:ImmunoOnly250707} shows the system when only immunotherapy is applied.  This is represented in the system as $\psi$, and we use $\psi = 302.5$. Again, the treatment is applied weekly for an hour at a time.  The results indicate that the cancer is eradicated on day 26, after three immunotherapy treatments.  

We complete the set of single treatments with virotherapy, represented in the system as $\varphi_V$.  The virotherapy is also applied once a week for an hour and is set as $\varphi_V = 350$.  Figure \ref{fig:ViroOnly250707} shows that the cancer disappeared on day 34, after four virotherapy treatments.

With these three basic separate treatments of the system established, one could test different combinations of treatments in an effort to maximize the effectiveness and minimize the discomfort to the patient, and other side effects.  

\begin{table}[H]
\label{tab:1}
\begin{minipage}{\linewidth}
\begin{tabular}{l*{6}{c}r}
Parameter    &  Value &  Description  \\
\hline
\hline
$q$ & 1.985 & rate of supply of immune cells\\
$\alpha_1$ & 0.6 & maximal growth rate constant of healthy cells\\
$\alpha_2$ & 1.55 & maximal growth rate constant of immune cells\\
$\alpha_I$ & 0.002 & reciprocal of immune cell carrying capacity\\
$\alpha_H$ & 0.001 & reciprocal of healthy cell carrying capacity\\
$\alpha_3$ & 1.1125 & maximal growth rate constant of tumor cells\\
$\alpha_{KU}$ & 0.001 & reciprocal of uninfected cell carrying capacity\\
$\gamma_1$ & 0.001 & loss rate of healthy cells due to uninfected tumor cells\\
$\gamma_2$ & 0.0107 & loss rate of immune cells due to uninfected tumor cells\\
$\gamma_3$ & 0.015 & loss rate of uninfected cancer cells due to  immune  cells\\
$\gamma_4$ & 0.001 & loss rate of uninfected cancer cells due to normal cells\\
$\mu_1$ & 0.001 & death rate of healthy cells due to drug\\
$\mu_H$ & 0.0005 & natural mortality rate of healthy cells\\
$\mu_2$ & 0.001 & rate at which the drug kills immune cells\\
$\mu_I$ & 0.01 & natural death rate of immune cells\\
$\mu_{KI}$ & 0.01 & natural death rate of infected immune cells\\ 
$\mu_3$ & 0.989 & death rate of uninfected cancer cells due to drug\\
$\mu_{KU}$ & 0.001 & natural death rate of uninfected cancer cells\\
$\mu_D$ & 0.1 & natural decay in vivo of the drug\\
$\mu_V$ & 0.1 & natural decay rate constant of viruses\\
$\delta_1$ & 0.88 & rate constant of infection of cancer by virus\\
$\delta_2$ & 0.001 & weight factor\\
$\delta_3$ & 0.0178 & death rate of infected cancer due to immune cells\\
$\delta_4$ & 0.0178 & death rate of infected cancer due to drug\\
$\beta_1$ & 0.01 & drug decrease rate due to immune cells\\
$\beta_2$ & 0.5 & drug decrease rate due to cancer cells\\
$\beta_3$ & 0.01 & drug decrease constant due to normal cells\\
$\beta_4$ & 0.1 & multiplier of viruses created by infected cells\\
\hline
\end{tabular}
\centering
\captionof{table}{Symbols and the values of the baseline parameters} \label{tab:1}
\vskip4pt
\end{minipage}
\vskip12pt
\end{table}

\begin{figure}[H]
\centering
    \includegraphics[scale = .34]{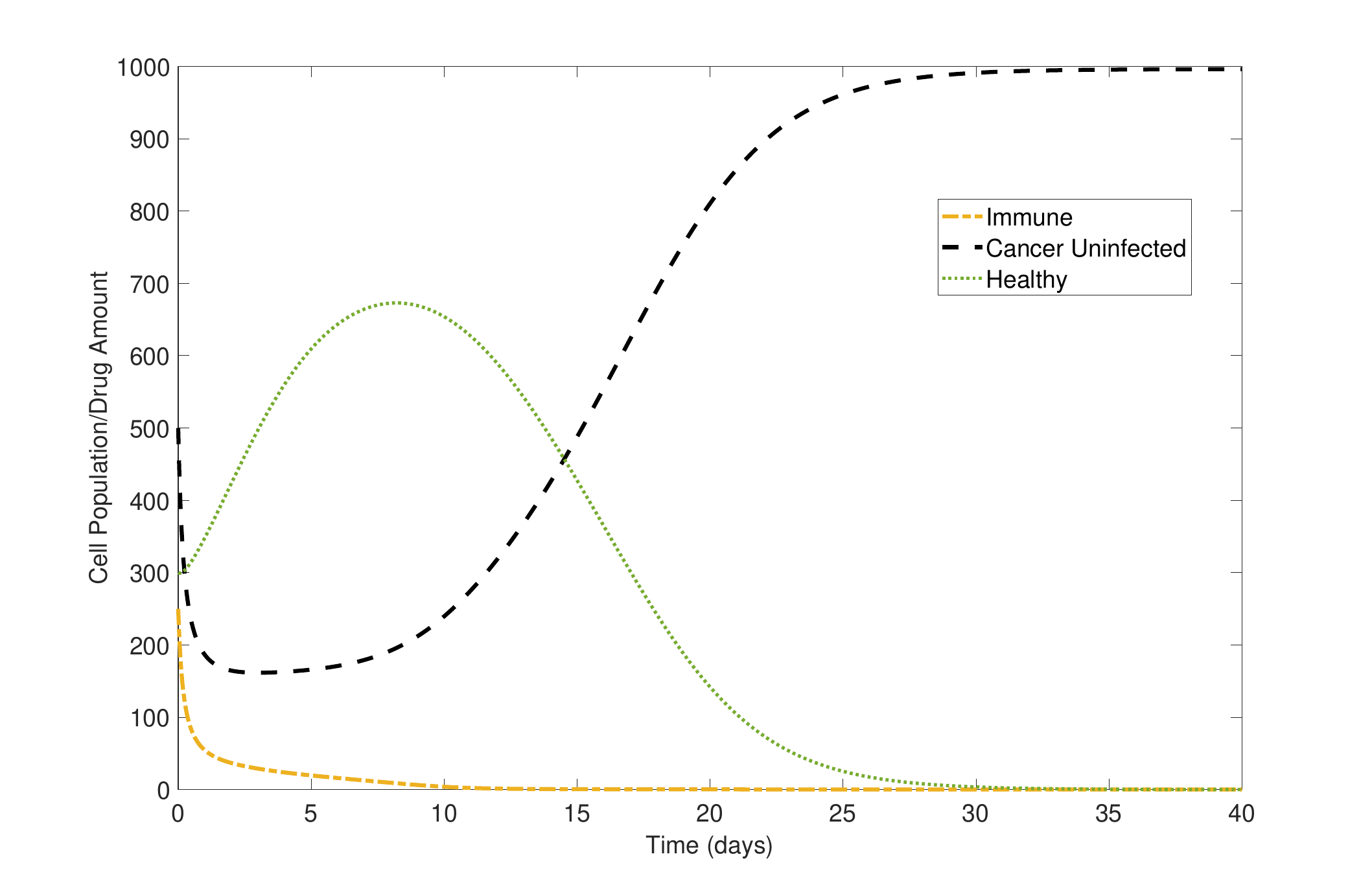}
    \caption{No Treatment. The behavior of the system without any treatment.  Data values are shown in Table 1. Cancer grows until it reaches its carrying capacity. Healthy and immune cells decline as the cancer grows, and without treatment, these cells eventually die out completely.}
    \label{fig:NoTreatment250707}
\end{figure}

\begin{figure}[H]
    \centering
    \includegraphics[scale = .34]{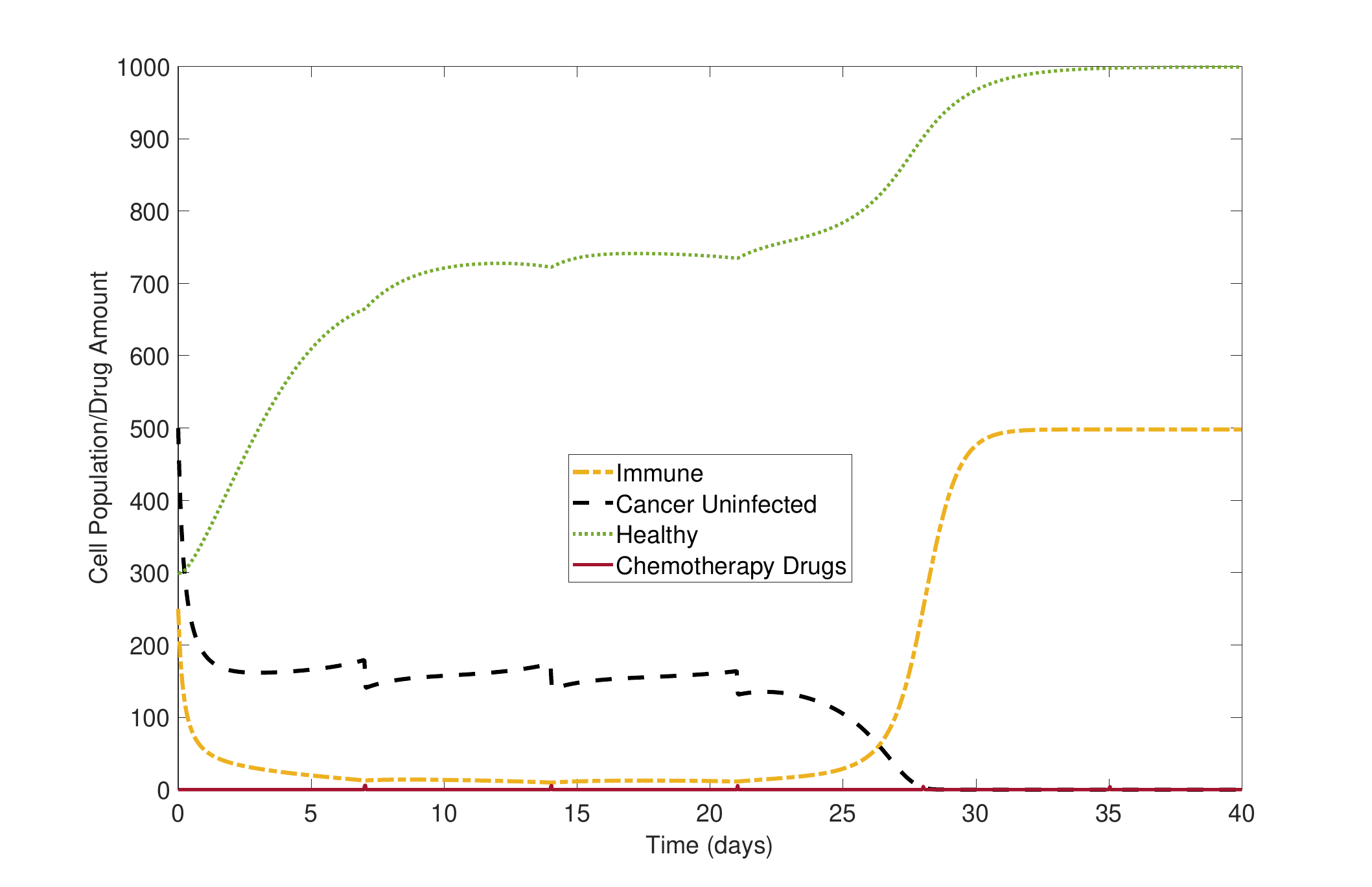}
    \caption{Chemotherapy Weekly. Only chemotherapy is applied in a weekly dose, where $\varphi_D = 2000$ for 1 hour of treatment on days 7, 14, and 21.  The graph indicates that after three rounds of chemotherapy cancer is eliminated on day 28. }
    \label{fig:ChemoOnly250707}
\end{figure}

\begin{figure}[H]
    \centering
    \includegraphics[scale = .35]{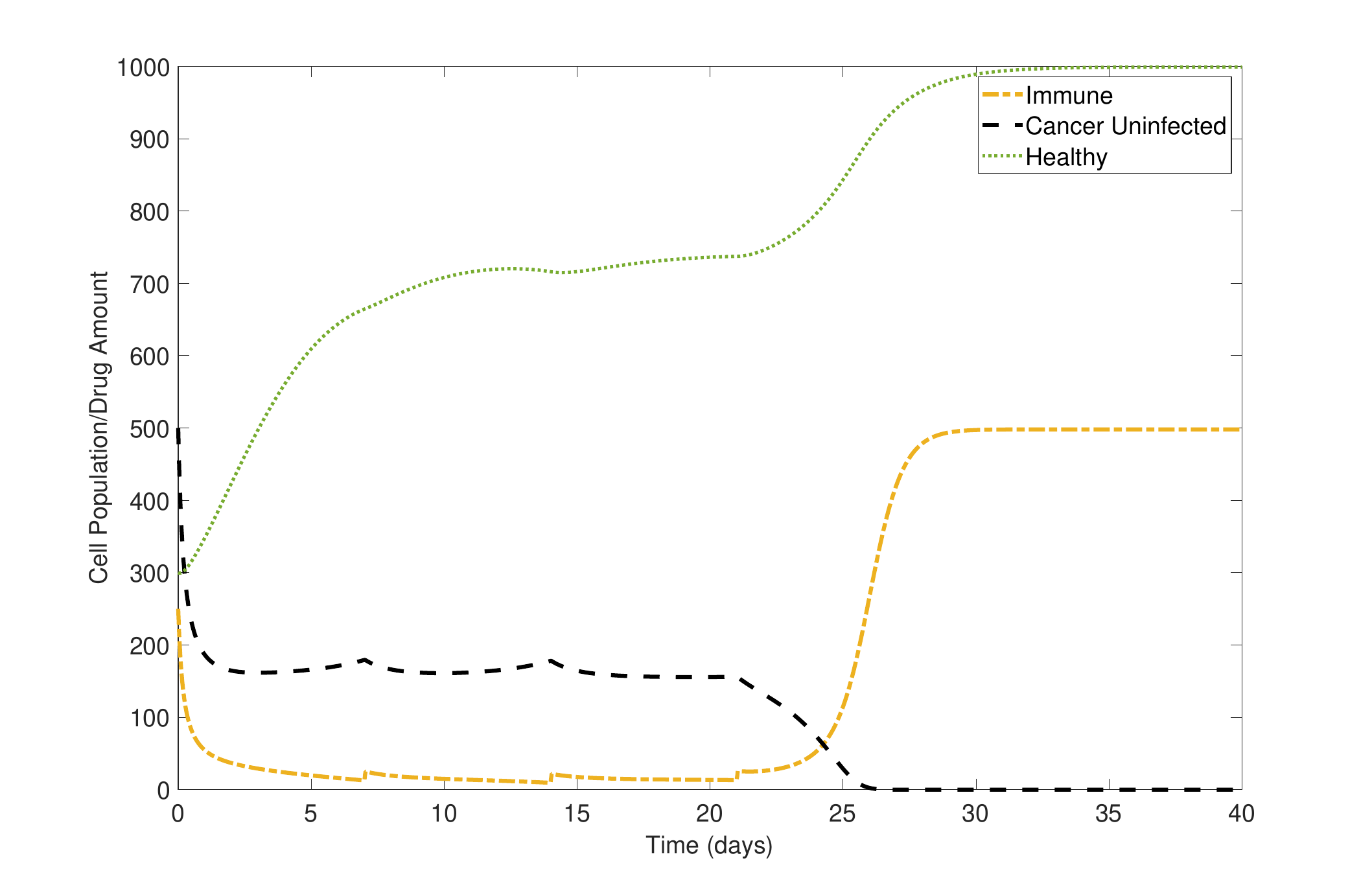}
    \caption{Immunotherapy Weekly. Only immunotherapy is applied in a weekly dose where $\psi = 302.5$ for 1 hour of treatment on days 7, 14, and 21.  Following three rounds of immunotherapy cancer is eliminated on day 26.}
    \label{fig:ImmunoOnly250707}
\end{figure}

\begin{figure}[H]
    \centering
    \includegraphics[scale = .35]{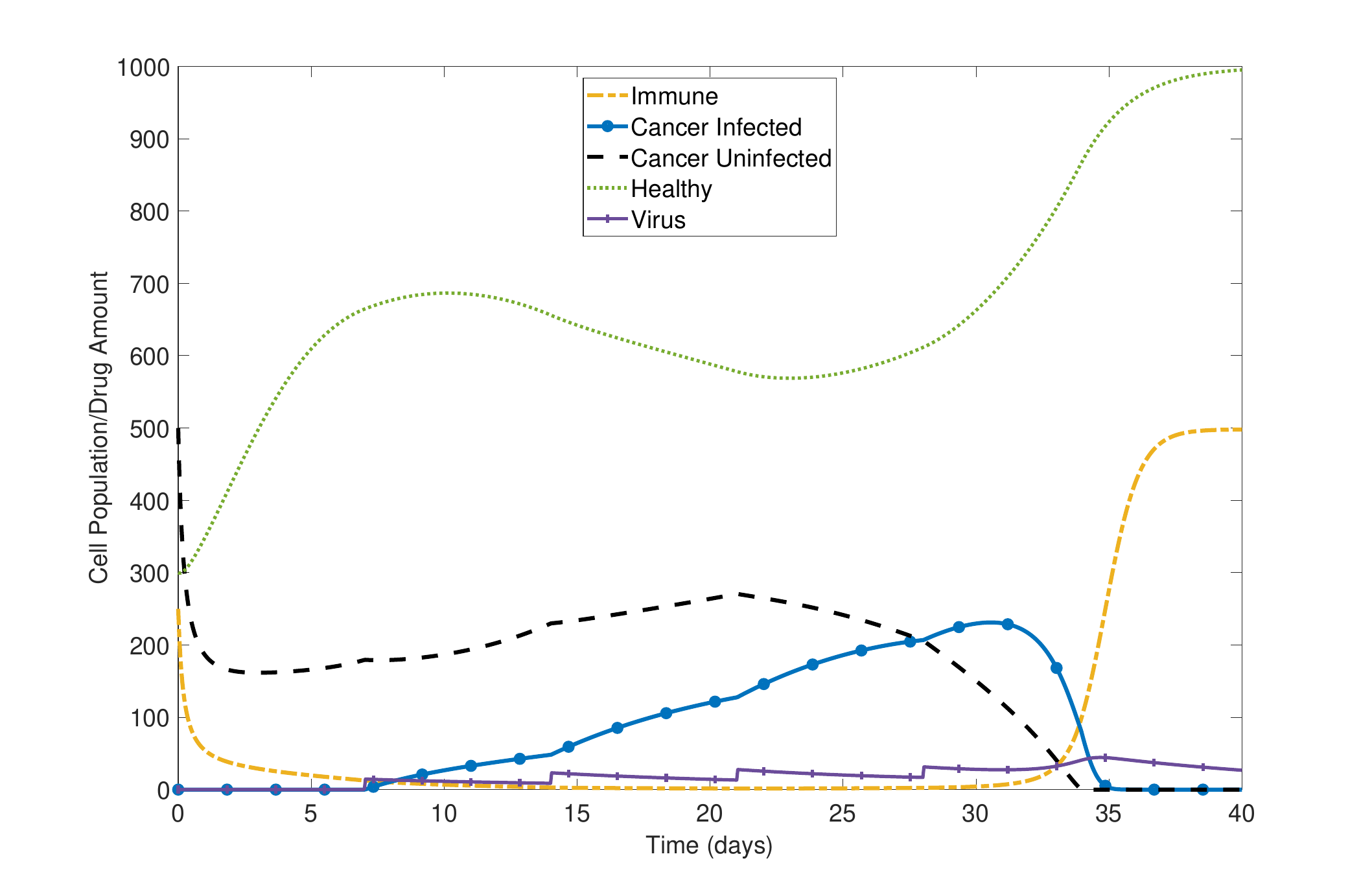}
    \caption{Virotherapy Weekly. Only virotherapy is applied in a weekly dose where $\varphi_V = 350$ for 1 hour of treatment on days 7, 14, 21 and 28.  It is seen that after four rounds of virotherapy cancer disappears on day 34.  }
    \label{fig:ViroOnly250707}
\end{figure}

\subsection{Combinations of treatments }
\label{sec:chemo} 

To further study the model and find better treatment outcomes, we consider the following computer experiment. We apply all three treatments together each week for a few weeks.  We show that applying the three treatments at once eradicates cancer faster than any single treatment.  This scenario is chosen to analyze and show the behavior of the system, not as a probable treatment schedule, as this would undoubtedly be too much for the patient's system to handle.  

We keep the treatment levels the same as is done for each treatment type: 

\[
\varphi_D = 2000, \qquad \psi = 302.5, \qquad \varphi_V = 350.
\]

The simulation results are shown in Figure \ref{fig:ThreeTreatmentsAtOnce250707}. After a single treatment, the cancer vanishes on day 11.  

\begin{figure}[h!]
    \centering
    \includegraphics[scale = .35]{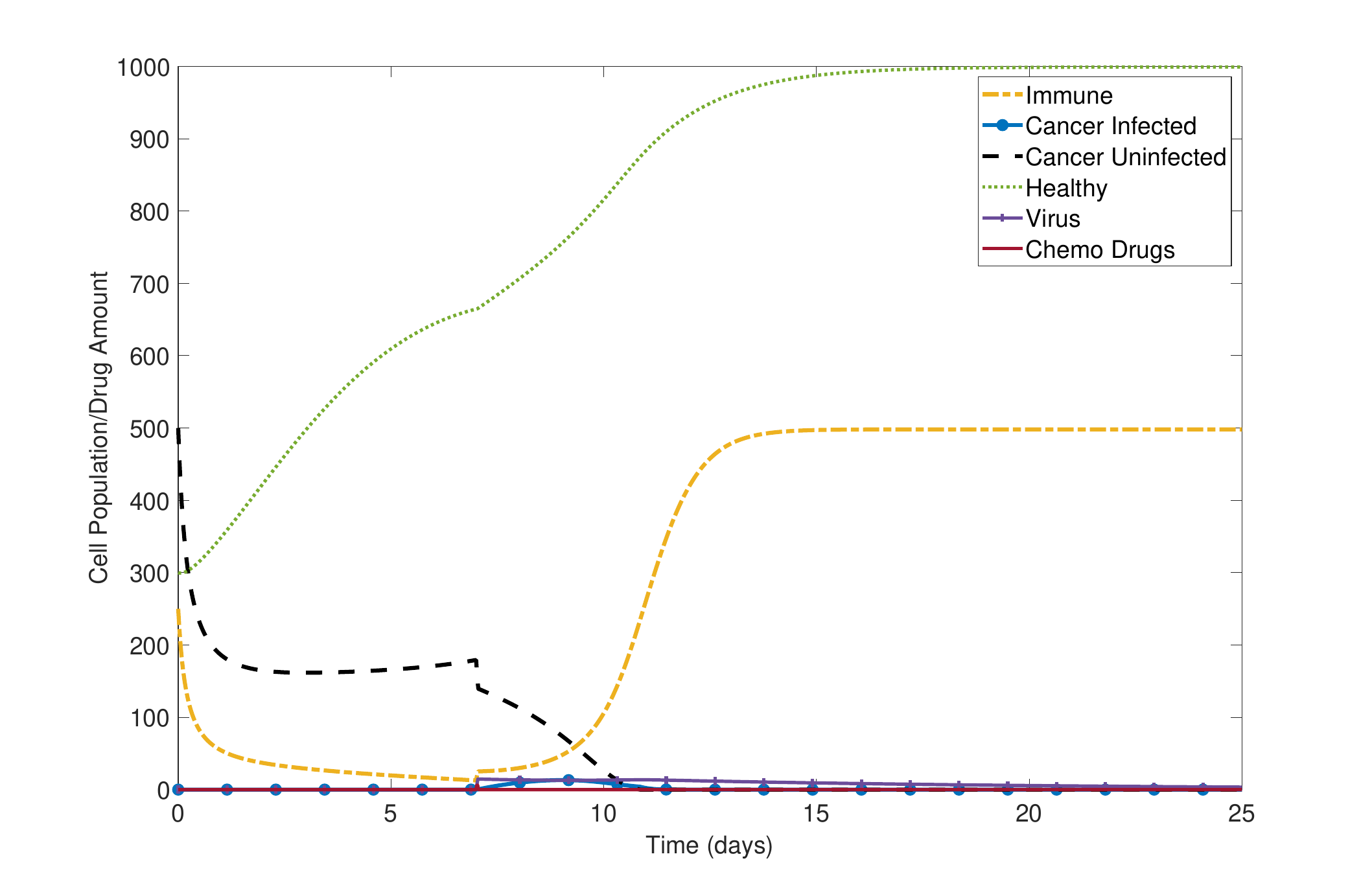}
    \caption{Three Treatments at Once. Weekly administered chemotherapy, immunotherapy, and virotherapy. Cancer (black -- uninfected cells) vanishes ($K_U<1$) after the first round, on day 11; then the healthy and immune cell populations increase monotonically towards their steady states.}
    \label{fig:ThreeTreatmentsAtOnce250707}
\end{figure}

To further study the model, we run simulations in which we applied the treatments in a cycle with three-day intervals. 
We use the same constants as given in Table \ref{tab:1}.  Figure \ref{fig:VIC2DaysBetween} shows the case where the first virotherapy treatment is applied on day 7, followed by immunotherapy on day 10, chemotherapy on day 13, and back to virotherapy on day 16. It is seen that cancer vanishes by day 17.  It appears from this graph that cancer is already declining following the second immunotherapy treatment.  In Figure \ref{fig:VIC2DaysBetween1RoundEach} we explore the results of only applying one round of virotherapy treatment and a round of immunotherapy treatment three days later.  The effect is that the cancer remains until day 18.  Although this is only a difference of a day, it illustrates that modifying treatments affects the rate at which cancer is eliminated. It also shows that in this case, just two treatments were enough the kill all the cancer cells. This may be of interest, since in this case only one round of virotherapy and one round of immunotherapy were used, which is likely to be very helpful to patients and save money and effort. We note that these results are similar to those in the first case of the combination of all three treatments.

\begin{figure}[h!]
    \centering
    \includegraphics[scale = .35]{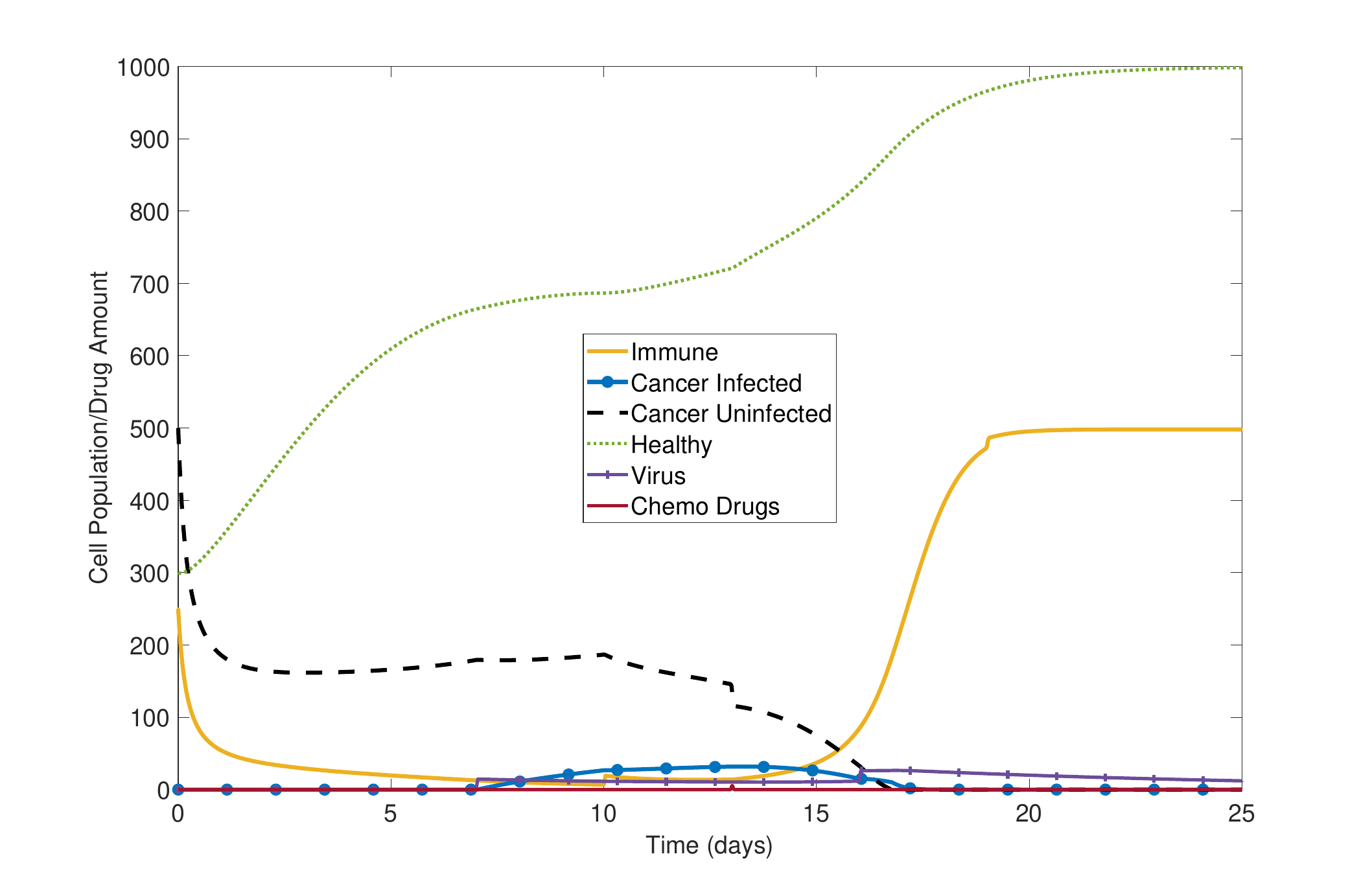}
    \caption{Treatments on a Three-Day Cycle. Treatments administered on a three-day cycle, alternating among virotherapy, immunotherapy, and chemotherapy. Cancer (black -- uninfected cells) vanishes ($K_U<1$) after the second round of virotherapy, on day 17.}

\label{fig:VIC2DaysBetween}
\end{figure}

\begin{figure}[h!]
    \centering
    \includegraphics[scale = .35]{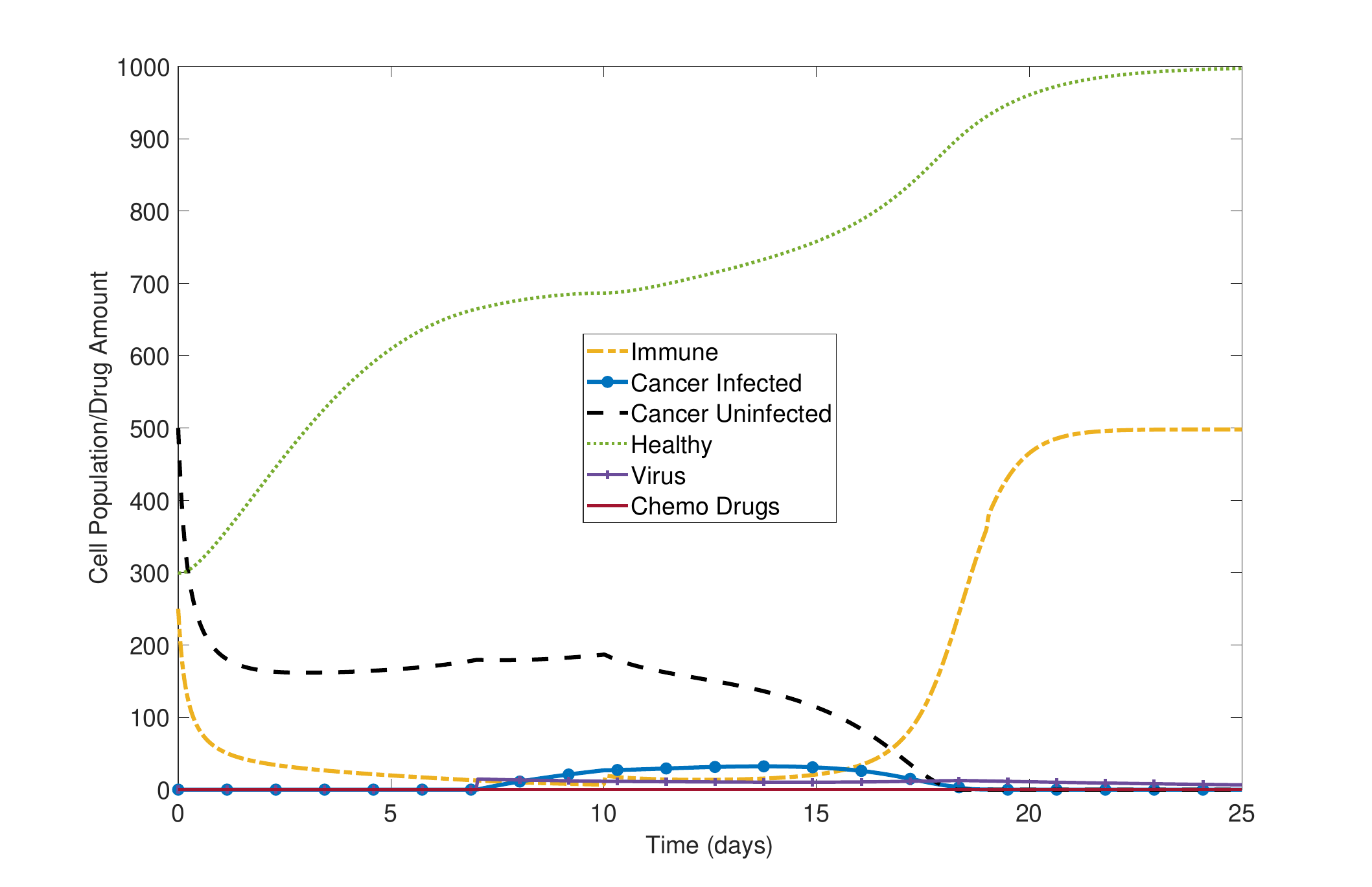}
    \caption{Treatments on a Three-Day Cycle, One Round Each. The three-day cycle is halted after the immunotherapy treatment to observe the effects of stopping treatment as soon as the cancer begins to decline. Compared to Figure \ref{fig:VIC2DaysBetween}, it takes a little over one additional day for cancer to vanish.}

\label{fig:VIC2DaysBetween1RoundEach}
\end{figure}

In other experiments with different combinations of the three treatments, but with the same treatment amounts, we obtained very similar results, in which the cancer vanished between 15 and 18 days. It seems that the model can predict the sequencing and time intervals well.

{\bf We note that some of the system coefficients are arbitrary, not based on data, and that the disappearance of the cancer in two or three weeks is the models artifact. However, if medically possible, these results indicate that combinations of treatments may be more effective than each one separately. }

\subsection{Numerical Convergence}
\label{sec:conv}
To check the convergence of the numerical algorithm, we run the case where virotherapy, immunotherapy, and chemotherapy are applied individually in that order on a three-day cycle, with the following time steps
\[
\Delta t= 10^{-4},\quad 5\times 10^{-5},\quad 10^{-5},\quad 5\times 10^{-6},\quad 10^{-6},\quad  5\times 10^{-7}.
\]
To show the numerical convergence of the model, we run the simulations with $\Delta t=5\times 10^{-7}$ and assume it to be the accurate solution. Then, we measured the maximal distance of the other numerical solutions from this solution. The results are given in Table \ref{tab:2}. 

\begin{table}[H]
\begin{minipage}{\linewidth}
\begin{tabular}{l*{7}{c}r}
$\Delta t$    &  $10^{-4}$ & $5\times10^{-5}$ & $10^{-5}$ & $5\times10^{-6}$ \\
\hline
\hline
$\|H_{\Delta t} - H_{\Delta t^{*}}\|_{\infty}$ & 5.39 (0.54\%) & 3.39 (0.34\%) & 0.50 (0.05\%) & 0.31 (0.031\%) \\
$\|I_{\Delta t} - I_{\Delta t^{*}}\|_{\infty}$ & 17.70 (3.6\%) & 11.12 (2.2\%) & 1.66 (0.33\%)& 1.00 (0.2\%)\\
$\|K_{I\Delta t} - K_{I\Delta t^{*}}\|_{\infty}$ & 2.78 (8.8\%) & 1.80 (5.7\%)& 0.28 (0.88\%) & 0.17 (0.54\%) \\
$\|K_{U\Delta t} - K_{U\Delta t^{*}}\|_{\infty}$ & 4.37 (0.87\%) & 2.75 (0.55\%) & 0.41 (0.082\%) & 0.25 (0.05\%)\\
$\|D_{\Delta t} - D_{\Delta t^{*}}\|_{\infty}$ & 0.20 (3.6\%) & 0.10 (1.8\%)& 0.02 (0.36\%) & 0.01 (0.18\%) \\
$\|V_{\Delta t} - V_{\Delta t^{*}}\|_{\infty}$ & 0.26 (0.98\%) & 0.13 (0.49\%)& 0.02 (0.075\%) & 0.01 (0.038\%) \\
\hline
\end{tabular}
\centering
\captionof{table}{Numerical convergence. The percentages given in the table compare the max norm to the maximum value that variable achieves using the reference time step, $\Delta t^{*}=5\times 10^{-7}$.   As the time step decreases, the difference of the solution at the given time step from the one with 
$\Delta t^{*}$ decreases notably. } \label{tab:2}
\vskip4pt
\end{minipage}
\vskip12pt
\end{table}
It is seen that as the time step decreases, the accuracy of the computations increases. Indeed, the maximum or $l_{\infty}$ norm decreases substantially. Indeed, comparing the error in the solutions with $\Delta t =10^{-4}$ with those with $\Delta t =10^{-5}$, or those with $\Delta t =5\times 10^{-4}$ with $\Delta t =5\times 10^{-5}$ seems to indicate that the scheme exhibits \underline{linear convergence}. This observation leads to our confidence in the numerical scheme and the simulations.

\section{Conclusions}
\label{sec:conclude}
We construct a dynamical systems model for the treatment of a cancer tumor using combinations of chemotherapy, immunotherapy, and virotherapy. The model consists of six ODE rate equations for Healthy, Immune, Uninfected Cancer, and Infected cell populations, and for the dynamics of the chemo drug and virotherapy.  In the model, we divide cancer cells into those infected by the virus, which are considered not dangerous because they cannot grow and multiply, and the uninfected, which are the disease. This model provides insight into and possible quantitative estimates of the effects of the three therapies administered alone or in combination. Its main goal is to be a tool for computer experiments that aim to improve treatment schedules.  When validated, the model can be used to optimize therapy schedules, therefore saving time, effort, and money, as well as reducing side effects while improving outcomes. Moreover, the model can be easily extended to include other treatment modalities and more complex and realistic settings. 

The model is presented in Section 2; its analysis is provided in Section 3, where the positive invariance, boundedness, and existence of the solutions can be found. The steady states of the system are studied in Section 4. It is seen that the DFE is stable and attracting, since all the eigenvalues of the Jacobian matrix are negative. However, in our setting, the stability of the DFE is of little interest, since the aim is to find conditions that {\it drive} the system to the DFE as fast as possible.

An important aspect of this work can be found in Section 5, where we construct and use a relatively simple explicit time-stepping finite-differences algorithm for the model simulations. The results show the dynamics of chemo-,  immuno-, or viro- therapies, and the conditions for the eradication of cancer.

In the simulations, we first study the case without medical intervention, Fig. 1, and as expected, the uninfected cancer cells grow to their carrying capacity, while the other cells die out. Then, in Figs. 2-4,  we study the cases where only chemotherapy,  only immunotherapy, or only virotherapy are applied. In all cases, therapies are applied weekly. The model predicts that in the chemotherapy case the uninfected cells die in 28 days, after three applications of the drug. In the immunotherapy case, they die in 26 days, again after three applications. In the virotherapy case, the uninfected cells die in 34 days, after four treatments. 

Fig. 5 shows the results of the application of all three therapies at the same time, and it is seen that the cancer is eradicated in 11 days, after a single application of each therapy. This is a very important prediction of the model.

To better understand how the model responds to different scenarios, in Figs. 6 and 7 we study a specific treatment regimen, alternating virotherapy, immunotherapy, and chemotherapy over three-day intervals, and examine the impacts of terminating  treatments before the cancer is gone.  When the cycle continues until the cancer is gone, it takes about 17 days for cancer to disappear.  When the treatment cycle is terminated as soon as the cancer appears to be declining, the cancer takes 18 days to be eliminated.  This indicates that the model has the potential to fine-tune the treatments.  Furthermore, the model predicts that the various combination of treatments result in the vanishing of cancer in about two weeks.

We conclude, tentatively, that in the model, whereas by itself chemotherapy eradicates cancer in about 28 days, immunotherapy in 26 days, and virotherapy in 43 days, their combined application eradicates it in 11 days. Similar results are obtained in the case where treatments are alternated. This result must be taken with caution, as the model is not yet validated.  This, when validated, can lead to considerable improvement in treatment outcomes in terms of patient care, substantial reduction in cost and time, better environment for patients, and overall use of medical facilities. {\bf We note, as above, that the model does no take into account any counter indications of using combinations of the treatments. However, the current medical practice does not mix them. Moreover, the model seem to overestimate how quickly the tumor vanishes, which is an artifact of the choice of some of the system parameters.}

It is found that as the time step decreases, the accuracy of the scheme simulations increases. It seems that the convergence is linear. However, the precise proof is unavailable, yet.

Furthermore, this work opens the door to a host of related models and their computer experiments. Moreover, some of the topics of interest are as follows.
\begin{itemize}
    \item Perform more extensive simulations of chemotherapy, immunotherapy, and virotherapy,  combination treatments and their interactions. In particular, with different schedules for the administration of the treatments.
    \item Add other possible treatments, modalities and methods, such as radiation, to the model.
    \item Perform parameter sensitivity analysis, to find out which are the important parameters.
\end{itemize}

We believe that the model, and its variants, will find ways to be applied in real world cancer treatments.

{\bf Statement on authors' contributions}. First author (TK Dutta): Conceptualization, methodology and some writing; Second Author: (SA Sangma): Some formal analysis and investigation. Third author: (J Moore): Software development, computer simulation, analysis and writing; Fourth author (M Shillor): Model development, analysis and writing.  

{\bf Declarations of interest}: none

{\bf Acknowledgment}. The fourth author (MS) is grateful to the Fulbright Foundation
for its support of his visit to Assam Don Bosco University (ADBU), Assam, India, during August 2024. He is also very grateful for the warm hospitality he experienced during his three-week stay there.

 The authors are very grateful to the two anonymous reviewers for the thorough reviews which improved the presentation.

\section{Appendix}
As noted in Section 4, we present the study of the endemic equilibrium (EE)  here. However, since the DFE is stable and attracting, we expect the EE to be unstable. We use the Jacobian matrix (\ref{J}). Let $(H^*,I^*,K_I^*,K_U^*,D^*,V^* )$ be the EE point. We are interested in the case when $K_U\geq 1$, hence  $\chi_1(K_U)=1$. Moreover, in the EE state,
there are no interventions, hence $\psi=\varphi_D=\varphi_V=0$. Then, the system is (\ref{41}),

The equilibrium system (omitting the stars), is the following:
\begin{eqnarray}
0&=&\alpha_1 H(1-\alpha_H H) -\gamma_1 K_UH-\mu_1 DH-\mu_{H}H,\notag \\
0 &=& q+\alpha_2I(1-\alpha_I I) - \gamma_2 K_UI -\mu_2 DI -\mu_{I}I,\notag  \\
0 &=& \frac{\delta_1 K_U V}{\delta_2 +K_U}  -\delta_3 K_I I -\delta_4 K_I D
-\mu_{K_I}K_I, \notag \\
0  &=&\alpha_3 K_U(1-\alpha_{K_U} (K_I+K_U))-\frac{\delta_1K_UV}{\delta_2+K_U} -\gamma_3 K_UI-\gamma_4 K_UH  \nonumber \\
&-&\mu_3 K_UD -\mu_{K_U}K_U, \label{A1}  \\
0 &=&-\beta_{1} ID-\beta_{2}HD-\beta_{3}(K_U+K_I)D-\mu_{D} D, \nonumber\\
0 &=& \beta_{4}(\delta_3 I +\delta_4 D)K_I -\mu_{V} V.  \nonumber
\end{eqnarray}

Next, the Jacobian matrix at the EE (stars omitted) is 
\[
J(H, I, K_I, K_U, D, V)  =
\]
\begin{equation}
\label{J}
\begin{bmatrix}
J_{11} & 0 & 0& -\gamma_1 H & -\mu_1 H &0 \\
0 &J_{22}& 0&-\gamma_2 I & -\mu_2 I&0 \\
0& -\delta_3K_I&J_{33}& \frac{\delta_1\delta_2 V}{(\delta_2+K_U)^2}&-\delta_4K_I&\frac{\delta_1K_U}{\delta_2+K_U}\\
-\gamma_4K_U& -\gamma_3 K_U& -\alpha_3\alpha_{K_U}K_U & J_{44}& -\mu_3 K_U & -\frac{\delta_1 K_U}{\delta_2+K_U} \\
-\beta_2 D & -\beta_1 D & -\beta_3 D & -\beta_3 D&J_{55}& 0\\
0&\beta_4\delta_3K_I,&\beta_4\delta_3I+\beta_4\delta_4D& 0&\beta_4\delta_4K_I& -\mu_V
\end{bmatrix},
\end{equation}

where the terms $ J_{11}-J_{55}$ are given by:
\begin{eqnarray*}
J_{11} &=& \alpha_1 (1 - 2\alpha_H H) - \gamma_1 K_U - \mu_1 D - \mu_H,\\
J_{22} &=& \alpha_2 (1 - 2\alpha_I I) - \gamma_2 K_U - \mu_2 D - \mu_I,\\
J_{33} &=&     - \delta_3 I - \delta_4 D-\mu_{KI} ,\\
J_{44} &=& (\alpha_3 -\alpha_3\alpha_{KU}K_I-2\alpha_3\alpha_{KU}K_U)
-\frac{\delta_1\delta_2V}{(\delta_2+K_I)^2} -\gamma_3I-\gamma_4H-\mu_3D-\mu_{KU},\\
J_{55}&=& -\beta_2H-\beta_1I-\beta_3(K_I+K_U)-\mu_D.
\end{eqnarray*}

The characteristic equation is 
\begin{equation}
\label{73}
\lambda^6+M_1 \lambda^5+M_2 \lambda^4+M_3 \lambda^3+M_4 \lambda^2+M_5 \lambda+M_6=0, 
\end{equation}

Where,

$M_1=-{\rm Trace}(J)=-(J_{11}+J_{22}+J_{33}+J_{44}+J_{55}+J_{66})$, 

$M_2$ is the sum of all principal $2\times 2$ minors of $J$, 

$-M_3$ is the sum of all principal $3\times 3$ minors of $J$, 

$M_4$ is the sum of all principal $4\times 4$ minors of $J$,

$-M_5$ is the sum of all principal $5\times 5$ minors of $J$,

$M_6=Det (-J)$.
\medskip

Now, we apply the Routh-Hurwitz criteria to find the stability conditions of the 
characteristics equation.
Routh-Array:

\[
\begin{array}{ccccc}
 \lambda^6 & 1& M_2& M_4 & M_6 \\
  \lambda^5   & M_1& M_3&M_5 &0 \\[2pt]
  \lambda^4 & b_1=\frac{M_1M_2-M_3}{M_1}& b_2=\frac{M_1M_4-M_3}{M_1}& M_6& 0\\[2pt]
  \lambda^3 & c_1=\frac{b_1M_3-M_1b_2}{b_1} & c_1=\frac{b_1M_5-M_1M_6}{b_1} &0&0\\[2pt]
  \lambda^2 & d_1=\frac{c_1b_2-b_1c_2}{c_1}&M_6 &0&0\\[2pt]
  \lambda^1& e_1=d_1c_2-c_1d_2&0&0&0\\[2pt]
  \lambda^0 &M_6&0&0&0
\end{array}
\]

According to the Routh criterion, the first column must all be positive, that is,
\begin{eqnarray*}
 && M_1>0, \\[2pt]
 && b_1=\frac{M_1M_2-M_3}{M_1}>0, \\[2pt]
 && c_1=\frac{b_1M_3-M_1b_2}{b_1}>0,\\[2pt]
 && d_1=\frac{c_1b_2-b_1c_2}{c_1}>0, \\[2pt]
 &&  e_1=\frac{d_1c_2-c_1M_6}{d_1}>0,\\[2pt]
 && M_6>0.
\end{eqnarray*}

We conclude that when all six conditions are satisfied, the EE is asymptotically stable.
If one of the conditions is reversed, the system remains unstable.
 
\end{document}